\documentclass{LNCS/llncs}

\title{How to Pick Your Friends}
\subtitle{A Game Theoretic Approach to P2P Overlay Construction}
\author{Saar Tochner\inst{1} \and Aviv Zohar\inst{1} \\
\email{ \{saart,  avivz\}@cs.huji.ac.il }}
\institute{The Hebrew University of Jerusalem, Israel}

\usepackage{xspace}
\usepackage{amsmath}
\usepackage{amsfonts}
\usepackage{mathrsfs}
\usepackage{graphicx}
\usepackage[outdir=graphs/]{epstopdf}
\usepackage{color}
\usepackage[]{algorithm2e}
\usepackage{caption}
\usepackage{pdfpages}

\usepackage{amsthm}

\begin{document}

\newcommand{\wattack}{W_{att}}
\newcommand{\wattacks}[1]{W_{#1}}
\newcommand{\gameval}{\mathcal{U}}
\newcommand{\easygameval}{\bar{\mathcal{U}}}
\newcommand{\nash}{Nash equilibrium\xspace}
\newcommand{\dif}{C}
\newcommand{\explain}[1]{\overset{\overset{#1}\downarrow}}
\newcommand{\updatefunc}{\mathbb{U}}
\newcommand{\rupdatefunc}{\updatefunc^{-1}}
\newcommand{\remfunc}{\mathscr{R}}
\newcommand{\buildfunc}{\mathscr{B}}
\newcommand{\setvalue}[1]{v_{\{#1\}}}
\newcommand{\sumxs}[1]{\sum_{i=#1}^{R}{x_i}}
\newcommand{\summs}[1]{\sum_{i=#1}^{R}{m_i}}
\newcommand{\masksShouldAdd}[1]{\frac{c_{new} + c_{node} \cdot t_U \cdot y_{#1} \cdot m_{#1}}{W({t_U}^{m_r} {y_{#1}}^{m_{#1}} e^H)^{y_{#1}} - W}}
\newcommand{\netcon}{net_{conn}}
\newcommand{\boldheader}[1]{\vskip 5pt \noindent{\bf #1}}
\newcommand{\avivz}[1]{{\color{red}avivz: #1}}
\newcommand{\todo}[1]{{\color{red}TODO: #1}}

\maketitle

\begin{abstract}
	A major limitation of open P2P networks is the lack of strong identities. This allows any agent to attack the system by creating multiple false personas, thereby disrupting the overlay network's connectivity and sabotaging its operation. In this paper, we explore practical ways to defend P2P networks from such attacks. To do so, we employ a game theoretic approach to the management of each peer's list of known nodes and to the overlay construction mechanisms that utilize this list. We consider the interaction between defender and attacker agents as a zero-sum game. We show that the cost of attacks can be driven up substantially if the defender utilizes available information about peers it chooses to connect to, such as their IP address. In addition to theoretical analysis of the underlying game, we apply our approach to the Bitcoin P2P network and derive effective strategies that guarantee a high safety level against attacks.
\end{abstract}

\section{Introduction}
P2P networks are utilized as an underlying communication layer in many applications such as Bitcoin~\cite{nakamoto2008bitcoin}, BitTorrent~\cite{norberg2009utorrent}, and DHTs~\cite{urdaneta2011survey}.
Unfortunately, they are extremely vulnerable to isolation attacks in which individual nodes that wish to participate in the network connect \textit{only} to attackers and are effectively quarantined from all other honest participants. Attackers can then distort the view of nodes regarding events in the network, filter their messages, change, or delay them at will.

Attacks of this sort give attackers great power. One prominent example is the Bitcoin protocol~\cite{nakamoto2008bitcoin}, in which isolated nodes can have their computational power subverted to attack the rest of the network, or may be blocked from issuing transactions.

In this paper we seek to make the attacker's job more difficult by suggesting strategies for the peers in the P2P network to successfully avoid being quarantined, unless the attacker invests a great deal of resources. In particular, our work models the interaction between the attacker and the defender as a game. The strategy space of the attacker is to choose the nodes to corrupt in the network (at a cost per node), while the defender's strategy space consists of the different sets of nodes he can connect to. Our model leads to \emph{practical} policies; We suggest ways to manage the limited buffer from which agents choose connections. Our work falls in the general category of security games: We consider the attacker and defender as rational agents that seek to maximize their utility.

The main idea that we utilize to increase the cost for the attacker is to take advantage of properties of the connections to peers to try and pick connections that the attacker would be unlikely to control all at once. 
In the main example, we take inspiration from Bitcoin's P2P formation and consider the IP subnet mask of the peers. Our main assumption is that once an attacker purchases an IP address in some range of IPs, it is cheap to gain access to other similar IPs. Our peer-selection strategy in this case will tend to be biased towards selecting peers from many different IP ranges. This affects the way we should manage the IP addresses' buffer. Our main result (Theorem~\ref{best_remfunc}) states that we can implement a safety-level strategy for this game, using only limited memory in which the IPs are stored. Our contribution is thus a highly practical one.

\boldheader{Current P2P network formation}
The mechanisms through which P2P clients choose their connections can vary, but most use the following general technique: Once an initial connection is established with a node (usually with the aid of a centralized server that helps to bootstrap the process), nodes share the IP addresses of others with their connections, and maintain lists of possible peers to connect to. Such potential connections are stored in buffers and often exchanged to keep the lists populated with ``fresh" addresses. Therefore, the core two decisions to be made are: (i) how to select which nodes to evict from a full buffer, and (ii) how to select nodes to connect to.

A naive approach is to form connections to uniformly selected nodes in the buffer (as such random graphs are typically well connected) and to evict nodes uniformly at random from an overflowing buffer. Other approaches involve removing the oldest IPs, assuming that they are the ones most likely to be stale. 

\boldheader{The susceptibility of naive strategies}
Attackers find it cheap to replace all the IPs in the buffer of the victim with IPs of attacker nodes (or with un-assigned IP addresses). The defender then fails to connect to other honest nodes and is quarantined by the attacker.

There are two approaches to an isolation attack: Sybil attacks~\cite{douceur2002sybil}, in which the attacker creates many identities to increase his chance of getting connections; and Eclipse attacks~\cite{singh2006eclipse}, wherein the attacker increases the visibility of each of his nodes according to the peers' selection algorithm (e.g., advertising its own IP more aggressively). If the defender receives many IP addresses belong to the attacker, and naively evicts nodes for his buffer in order to store these new addresses, then the buffer will be overrun and the node will not be able to connect to any honest peers.

Bitcoin employs a more sophisticated eviction strategy, sorting IP addresses into buckets by combining the IP of the sender node with the IP address in the message itself. Recent work shows this mechanism to be susceptible as well, as it still allows attacker to overrun the buffer easily~\cite{heilman2015eclipse}.

The remainder of the paper is structured as follows: In the next section we briefly review related work, then in section~\ref{sec::model} we define utility functions for the players, the strategy space and the agents' memory buffer. Section~\ref{sec::analysis_unrestricted_buffer} shows results for the game without any restrictions on the buffer size. In section~\ref{sec::buffer_restricted_imp} we restrict the size of the memory buffer and derive the safety-level strategies for the defender. Section~\ref{sec::costs_by_ip_mask} addresses a specific implementation of the cost function and contains our main practical theorem. Finally, subsection~\ref{sec::evaluations} evaluates our method as implemented on the Bitcoin network, using real data that we crawled, and compares the result to other benchmarks. We conclude and discuss future work in section~\ref{sec::future_work}.

\subsection{Related Work}
Douceur~\cite{douceur2002sybil} was the first to expose the problem of multiple identities (Sybils) in P2P systems lacking strong identities. Many such systems have since been shown to be vulnerable to such attacks~\cite{urdaneta2011survey}.

Bitcoin, a P2P currency system~\cite{nakamoto2008bitcoin}, was designed to work as an open system with no strong identities. Its overlay formation is thus susceptible to attacks~\cite{heilman2015eclipse} \cite{apostolaki2016hijacking}.
Botnets are often structured like P2P networks, and their susceptibility to such attacks has been used to attack the botnets themselves~\cite{jiang2014detecting,jung2016security}.

One approach to defend from Sybils was to utilize a social network to form the peers' connections~\cite{xinhui2015sybil,yu2006sybilguard} (Sybils are assumed to have few connections to honest participants in this setting). Unfortunately, most settings do not have this additional network of relations among peers to use for overlay formation.
Another approach is to make it feasibly hard to create many entities by having to solve a computational puzzle~\cite{zhang2011making}, but this could be difficult to accomplish in different networks, if the nodes lack computational power (e.g., mobiles).

The approach of modeling the interaction between attackers and defenders using game theoretic tools is well established
~\cite{manshaei2013game}. A series of papers on security games deals with several variants of such problems. These include the ARMOR project for security at airports~\cite{pita2008deployed}, resilience to cyberphysical control systems~\cite{zhu2015game}, and information security~\cite{fielder2014game}.

\section{Model} \label{sec::model}
We assume the attacker wishes to separate the defender node from all other honest nodes. We model this as a game between the defender and the attacker. The first model that we examine assumes (unrealistically) that the defender knows all the nodes in the ``universe'' $V$ and is unrestricted by memory considerations. We later use this as a building block for the second model in which the defender has a limited knowledge and memory buffer.

It is important to note that our model was created to defend from attacks on a single node. This is not the general case, where the attacker may try to isolate a specific subset of nodes. In this paper we will approximate an honest node as ``safe'' simply if it is connected to other honest nodes in the local sense.

\subsection{Preliminaries}
Let $S^1,S^2$ be the strategy sets of the defender and the attacker, respectively. Below we define  $\gameval: S^1 \times S^2 \rightarrow \mathbb{R}$ to the defender's utility in the game (the attacker's utility is $-\gameval$, this is a 2-player zero-sum game).
Recall that strategy $s^1 \in S^1$ is called an \emph{L-safety-level strategy} iff $\forall s^2 \in S^2, \, \gameval(s^1, s^2) \geq L$ (a similar definition can be stated for player 2).

Denote V as the set of peers (either honest or those owned by the attacker) in the network.
Let $H$ be the number of connections the defender creates (as in Bitcoin implementation, we use a configurable constant).

We further assume that the attacker gains some value $\wattack$ from a successful attack against the defender (and that the defender suffers this as a loss).

\subsubsection{Cost Functions} \label{subsec::costs}
We assume the attacker can corrupt or \emph{acquire} a subset $A$ of nodes from the universe $V$ at a cost that we denote by $\dif(A)$, where $\dif : 2^V \to \mathbb{R}_{\ge 0}$.
Intuitively there are many choices for such functions (constant cost, costs that are relative to the round trip time and many more). In this paper, for simplicity, we mainly focus on a specific cost function, in which acquiring an IP range costs $c_{new}$ and every IP from that range costs $c_{node}$, assuming $\dif(A) = c_{new} \cdot \text{number\_of\_masks\_in\_A} + c_{node} \cdot |A|$.

We start the discussion with the general cost function, where we prove some interesting theoretical theorems. After that, in section~\ref{sec::costs_by_ip_mask}, we use the IP range costs, conclude the main practical theorem, and evaluate the results.

\subsection{The Limitless-Buffer Game}
The 2-player zero-sum game between attacker and defender is defined as follows:
The defender's strategy space $S^1$ contains possible sets of nodes that he may connect to: $S^1 = \{A \subset V : |A|=H\}$.
The attacker's strategy space $S^2$ contains subsets of the universe $V$ that he chooses to corrupt: $S^2 = 2^V$.
The attack is considered successful iff the attacker owns all the nodes that the defender selected. The utility function is thus $\gameval(A, B) =
	\begin{cases}
	\dif(B) & \text{ if } A \not \subset B \\
	\dif(B) - \wattack & \text{ if } A \subset B
	\end{cases}$

We will usually consider the game with mixed-strategy $\sigma^1 \in \Delta_{S^1}, \sigma^2 \in \Delta_{S^2}$:
	$$\gameval(\sigma^1, \sigma^2) = \sum_{B \in S^2}{\sigma^2_{B} \dif(B)} - \wattack\sum_{B \in S^2, A \in S^1}{\sigma^1_{A} \sigma^2_{B} \delta_{A \subset B}} $$  in the above $\sigma^k_U$ is the probability that player $k$ chooses the subset $U \in S^k$, and $\delta_{A \subset B}$ is 1 if $A \subset B$, 0 otherwise.

The utility function is in fact the expected money that the attacker will \textit{pay} if he tries to attack. Therefore, the defender tries to gain a positive utility, because this protects him from being attacked. The attacker tries to get negative utility because this is the total value that he will earn.

\subsection{The Restricted-Buffer Scenario} \label{restricted_buffer_subsection}
We wish to defend P2P networks in realistic settings so we must take buffers and limited knowledge of the world into account. We thus define the following refinement of the game: Assume that at any point in time the defender can only maintain a set of potential peers in his buffer. Let $\mathcal{B}$ denote the buffer size (the number of nodes for which information can be saved). The agent receives a stream of announcements about nodes from which it selects ones to store in the buffer. We assume honest nodes are advertised at least once in each time period of length $\mathcal{T}$ (attacker nodes may be updated more frequently, as is often the case during eclipse attacks). If a node's details are stored in the buffer, and the buffer is full, a different stored record must be evicted first. The node then chooses its connections from the set of stored nodes. 

Each node should specify an algorithm $\remfunc$ that implements the functionality that decides, based only on the cost to corrupt the nodes and unique identity (the only revealed information), which node records to save or evict from the buffer.
As the defender, we try to find $\remfunc$ that maximizes the minimum amount of resources that must be invested to successfully attack.

\section{Analysis with an Unrestricted Buffer}\label{sec::analysis_unrestricted_buffer}
In this section we show results for the setting in which the nodes in $V$ are known to all.
The following theorem shows that it is better for the defender to err on the side of over-estimating the damage from an attack:
\begin{theorem} \label{smaller_worths}
Let $\gameval_{\wattacks{i}}$ be the utility of the game in which a successful attack causes $\wattacks{i}$ damage. If strategy $\sigma^1$, is a l-safety-level in the game with $\wattacks{1}$ and $W_2 \le W_1$, then $\sigma^1$ is a l-safety level in $\gameval_{\wattacks{2}}$, and generally:
$$\forall \bar{\sigma}_2 \quad \gameval_{W_2}(\sigma^1, \bar{\sigma}^2) \ge \gameval_{W_1}(\sigma^1, \bar{\sigma}^2) \ge l$$
\end{theorem}
\begin{proof}
	$\gameval_{\wattacks{2}}(\sigma^1, \bar{\sigma}^2) = 
	\sum_{B \in S^2}{\bar{\sigma}^2_B \dif(B)} - \wattacks{2}\sum_{B \in S^2, A \in S^1}{\sigma^1_A \bar{\sigma}^2_B \delta_{A \subset B}} \ge$
	$$ \sum_{B \in S^2}{\bar{\sigma}^2_B \dif(B)} - \wattacks{1}\sum_{B \in S^2, A \in S^1}{\sigma^1_A \bar{\sigma}^2_B \delta_{A \subset B}}  = 
	\gameval_{W_1}(\sigma^1, \bar{\sigma}^2) \ge l $$
	The first step is using the definition, the second is $W_2 \le W_1$, and the last step is because $\sigma^1$ is a l-safety-level in $\gameval_{\wattacks{1}}$.
\end{proof}

Next, we observe the attacker's strategies in Nash equilibria. We show that an attack may come in one of two forms: either the attacker corrupts no nodes at all or he places some small probability to ``cover'' every node.
\begin{lemma} \label{attacker_play_all}
	For any \nash $(\sigma^1, \sigma^2)$ in the game it holds that either $\forall B \quad \sigma^2_B = 0$, or alternatively, $\forall A \in S^1~ \exists B$ s.t. $A \subseteq B$ and $\sigma^2_B \ne 0$.
\end{lemma}
\begin{proof}
	Mark the defender's strategy to play pure $A \subset V$ with $S^1_A$.
	Let us now assume that exists $A \in S^1$ s.t. for all $B$ with $A \subseteq B$ it holds that $\sigma^2_B = 0$, then the attacker never answers the defender's strategy $S^1_A$, therefore $\gameval(S^1_A, \sigma^2) \ge 0$. But it follows by the \nash definition that $\gameval(\sigma^1, \sigma^2) \ge \gameval(S^1_A, \sigma^2)$.
	Clearly $0 = \gameval(\sigma^1, 0) \ge \gameval(S^1_A, \sigma^2)$ therefore $\gameval(\sigma^1_A, \sigma^2) = 0$.
\end{proof}

We now show that in any \nash, the defender places more probability on selecting more expensive sets of nodes (from the attacker's support).
\begin{lemma} \label{difficulty_lemma} Let $(\sigma^1, \sigma^2)$ be a \nash. Then $\forall {B_1},{B_2} \in supp(\sigma^2)$ it holds that $\dif({B_1}) \le \dif({B_2})$ if and only if $\sum_{A \subseteq {B_1} \cap S^1} \sigma^1_A \le \sum_{A \subseteq {B_2} \cap S^1} \sigma^1_A$
\end{lemma}
\begin{proof}
	Reminder: If v is the value of the game A and $\left(\sigma^1, \sigma^2\right)$ is an equilibrium, then any pure strategy from the support can be used to achieve it. Formally, if $\sigma^2_B \ne 0$ (for some $B \subset V$), then $\gameval \left(\sigma^1, S^2_B\right) = v$. Therefore:\\
	$v = \gameval\left(\sigma^1, S^2_{B_1}\right) = \dif({B_1}) - \wattack \cdot \sum_{A \subset {B_1} \cap S^1} \sigma^1_A$ \\
	$v = \gameval\left(\sigma^1, S^2_{B_2}\right) = \dif({B_2}) - \wattack \cdot \sum_{A \subset {B_2} \cap S^1} \sigma^1_A$ \\
	The first step in each row is from ${B_1},{B_2} \in supp(\sigma^2)$, the second is from the definition of $\gameval$.
	Therefore
	$ \dif({B_1}) - \dif({B_2}) = \wattack \left( \sum_{A \subset {B_1} \cap S^1} \sigma^1_A - \sum_{A \subset {B_2} \cap S^1} \sigma^1_A \right) $.
	So indeed $\dif({B_1}) \le \dif({B_2}) \leftrightarrow \displaystyle{\sum_{A \subset {B_1} \cap S^1}} \sigma^1_A \le \displaystyle{\sum_{A \subset {B_2} \cap S^1}} \sigma^1_A$.
\end{proof}

The next corollary is obvious, but it contains an insight that we will use in the next sections:
\begin{corollary} \label{bigger_difficulty_corollary}
	The probability of choosing a group of peers is determined by the attacker's support.
	I.e., let $(\sigma^1, \sigma^2)$ be an equilibrium, then $\forall B_1,B_2 \in supp(\sigma^2)$, it holds that $\dif(B_1) = \dif(B_2)$ if and only if $\sum_{U \subseteq B_1 \cap S^1} \sigma^1_U = \sum_{U \subseteq B_2 \cap S^1} \sigma^1_U $
\end{corollary}

\section{Buffer-Restricted Implementation} \label{sec::buffer_restricted_imp}

We now turn to the scenario where buffers are restricted.
Our goal here is to decrease the number of nodes that we should remember in the buffer, while achieving {\it the same game value}. We do so by changing the defender's strategies space using equivalence classes on $S^1$ and then prove that it preserves the \nash and the game value.

Define the equivalence relation $A_1~\sim~A_2$ on sets of nodes $A_1,A_2 \in S^1$ if $\forall (\sigma^1, \sigma^2)$ \nash in $\gameval$, holds that $\sigma^1_{A_1} = \sigma^1_{A_2}$. We will denote this equivalence class with $[A_1]$.
Using those equivalence classes as the defender's strategies space ($\hat{S^1}$) we can define a new game ($\hat{\gameval}$).

Intuitively, in this game, the defender does not distinguish between different connections' sets within the same equivalence class, so he chooses one uniformly. This definition can be also related to the cost to corrupt sets of nodes in any attacker's Nash strategy (as in Corollary~\ref{bigger_difficulty_corollary}).

Due to lack of space, we will not present here the full definition of $ \hat{\gameval}$ and the analogy between the two perspectives. In a nutshell, $\hat{\gameval}([A], B)$ is defined as choosing uniformly strategy $A_1 \in [A]$ and playing it in $\gameval(A_1, B)$. After that, we created a mapping function $T: \Delta_{S^1} \to \Delta_{\hat{S^1}}$ with $\big( T(\sigma) \big)_{[A]} = \sum_{b \in [A]} \sigma_b$ and $\big( T^{-1}(\hat{\sigma^1}) \big)_{A} = \frac{\hat{\sigma^1}_{[A]}}{|[A]|}$. We prove that for strategies $(\hat{\sigma^1}, \sigma^2)$ it holds that (i) $ \hat{\gameval}(\hat{\sigma^1}, \sigma^2) = \gameval(T^{-1}(\hat{\sigma^1}), \sigma^2)$, (ii) $T(T^{-1}(\hat{\sigma^1})) = \hat{\sigma^1}$, and (iii) if $(\sigma^1, \sigma^2)$ is \nash in $\gameval$, then $\forall \sigma'^2 \in S^2$ it holds that $\hat{\gameval}(T(\sigma^1), \sigma'^2) = \gameval(\sigma^1, \sigma'^2)$.

Finally, we conclude the main theorem:

\begin{theorem} \label{reduce_theorm}
	Nash equilibria in both games have the same game value, and  $(\sigma^1, \sigma^2)$ is a \nash in $\gameval$ iff $(T(\sigma^1), \sigma^2)$ is a \nash in $\hat{\gameval}$.
\end{theorem}
\begin{proof}
	We first show that the $T$ function saves the property of \nash.
	Let $(\sigma^1, \sigma^2)$ be a \nash in $\gameval$, we want to prove that $(T(\sigma^1), \sigma^2)$ is \nash in $\hat{\gameval}$. 
	On the one hand, $\forall \sigma'^1 \in \hat{S^1}$, 
	$$ \hat{\gameval}(\sigma'^1, \sigma^2) = \gameval(T^{-1}(\sigma'^1), \sigma^2) \le \gameval(\sigma^1, \sigma^2) = \hat{\gameval}(T(\sigma^1), \sigma^2) $$
	where the first equality is (i), the second is the definition of \nash in $\gameval$, and the last is (iii).
	On the other hand, $\forall \sigma'^2 \in S^2$
    $$ \hat{\gameval}(T(\sigma^1), \sigma'^2) = \gameval(\sigma^1, \sigma'^2)  \ge \gameval(\sigma^1, \sigma^2) = \hat{\gameval}(T(\sigma^1), \sigma^2) $$
	where the first and last equality are (iii), and the second is the definition of \nash.
	
	Then we prove that this is a \nash. And clearly, they have the same value due to (iii): $\gameval(\sigma^1, \sigma^2) = \hat{\gameval}(T(\sigma^1), \sigma^2)$
\end{proof}

\begin{corollary} \label{equivalnece_safety_level_corollary}
	Defender's safety level l in $\hat{\gameval}$ can be translated to safety level $\ge l$ in $\gameval$.
\end{corollary}
\begin{proof}
	Let $\hat{\sigma}^1 \in \hat{S^1}$ be a safety level l in $\hat{\gameval}$. Indeed, $\forall \sigma^2 \in S^2$, $\gameval(T^{-1}(\hat{\sigma^1}), \sigma^2) = \hat{\gameval}(\hat{\sigma^1}, \sigma^2) \ge l$
	where the equality on the left is (i), and the inequality on the right is the definition of safety level l in $\hat{\gameval}$.
\end{proof}

\boldheader{Equivalent peers}
{\ }For reasons that will become clear below, define a new equivalence relation on $V$: Two nodes $u,v$ are equivalent if for any set of other nodes $A$, adding $u$ to $A$ results in a strategically equivalent set to $A \cup \{v\}$. Formally: $\forall v,u \in V$, mark that $u \rightleftharpoons v$ if $\forall A \subset V, |A| = H-1$ it holds that $A \cup \{u\} \sim A \cup \{v\}$. Let $[v]_\rightleftharpoons$ denote the equivalence class of node $v \in V$, and the equivalence classes space with $\hat{V} := \{ [v]_\rightleftharpoons | v \in V \}$. 

Note that $[\{v_1, \cdots, v_H\}] \ne [v_1]_\rightleftharpoons \times \cdots \times [v_H]_\rightleftharpoons$ and generally there is no containment in either direction
(see Lemma \ref{equivalence_lemma_masks} for a specific interesting property).

\subsection{The Buffer Management Algorithm}
In this section we propose a concrete way of implementing the buffer management algorithm $\remfunc$ (as described in subsection~\ref{restricted_buffer_subsection}), that utilizes a Bloom filter. 

A Bloom filter~\cite{bloom1970space} is a data structure that uses hash-coded information to encode a set of items. It allows for some small fraction of errors in membership tests (false positives). The error rate can be lowered by increasing memory usage. In our paper, we use this method to avoid nodes that we already saw including those evicted from the buffer. We use the Bloom filter's deterministic answer to completely avoid known nodes, and accept a small fraction of false-positive answers on new nodes.

The next theorem discusses the benefits of well-implemented buffer management algorithms, and the following lemma gives a possible implementation:	

\begin{theorem} \label{buffered_game_value}
	Assume that we have some buffer management algorithm $\remfunc$ that satisfies the conditions from subsection \ref{restricted_buffer_subsection}, and in addition it has the property that choosing uniformly a node from $[A]$ inside the limited buffer has the same probability as choosing it uniformly from $[A]$ inside the entire inputs. \par
	Then we can implement any Nash strategy or safety level on a restricted buffer of size $O(H \cdot |\hat{S^1}|)$ with the same value as the game on a limitless buffer.
\end{theorem}
\begin{proof}
	Store $H$ nodes for any equivalence class in $\hat{S^1}$. 
	Using the described property, $\forall [A] \in \hat{S^1}$, choosing a connection set uniformly in $[A]$ from the set of nodes in the buffer is the same as choosing it uniformly from the entire input. Therefore, we can play any strategy $\hat{\sigma}^1$ in the game $\hat{\gameval}$ by uniformly selecting a group in the buffer that is in the chosen equivalence class.
	
	Therefore, we implement a choice that is equivalent to the defender's strategy space $\hat{S^1}$ in the game $\hat{\gameval}$.
	Let $\sigma^2$ be an attacker strategy. The value of the game that was played in this buffer-limited world is exactly $\hat{\gameval}(\hat{\sigma}^1, \sigma^2)$. Finally, if this is a \nash / safety-level strategy, then $\gameval(T^{-1}(\hat{\sigma}^1), \sigma^2)$ is a \nash / safety-level in $\gameval$ (Theorem~\ref{reduce_theorm}/Corollary~\ref{equivalnece_safety_level_corollary} resp.).
\end{proof}

Consider Algorithm~\ref{remfunc_imp} with:
$EG$ as the equivalence classes (``buckets'') in our game, $bucket\_size$ as the number of nodes that each bucket can hold, $B \in \mathbb{N}$ as the size in bytes of the Bloom filter, and $Prob(p)$ to be true with probability $p$ (if $p \ge 1$ it is always true).

\begin{algorithm}
	\caption{$\remfunc$ algorithm for EG} \label{remfunc_imp}
	Initialize: \\ 
		~~~BF := Bloom-filter buffer of size $B$. \\
		~~~$\forall$ bucket $\in$ EG: bucket\_history[bucket]=0 \\
		~~~$\forall$ bucket $\in$ EG: buckets[bucket]=$\emptyset$ \\
	\While{n := new input node}{
		\If{not BF.contains(n)} {
			b := n.bucket \\
			bucket\_history[b] ++ \\
			\If{Prob$(\frac{bucket\_size}{bucket\_history[b]})$} {
					\If{buckets[b].isFull} {
						buckets[b].uniformlyRemoveOne
					}
					buckets[b].add(n)
			}
			BF.add(n)
		}
	}
\end{algorithm}

\begin{lemma} \label{remfunc_implementing_bloom_filter}
	Algorithm~\ref{remfunc_imp} with $EG = \hat{S^1}$ and optimal filter\footnote{filter without false positive or false negative.} satisfies the conditions of Theorem \ref{buffered_game_value}.
\end{lemma}
\begin{proof}
For any equivalence class in the defender's strategy space (``bucket"), we should prove the following claim: If we choose a node uniformly from the bucket that we stored in the buffer, then it has the same probability as choosing it uniformly from all the nodes that the defender is notified about and belong to this equivalence class (i.e., the case where the bucket size is infinite). Assume that we make the choice after being notified of $l$ nodes in this bucket. We need to prove that there is a chance of $\frac{1}{l}$ to choose any node. 
Indeed, $\forall j \in \{1, \cdots, l\}$, we insert it into the buffer with probability $\frac{bucket\_size}{j}$ and after the next input it is still there with probability $(1 - \frac{bucket\_size}{j+1} \frac{1}{bucket\_size}) = \frac{j}{j+1}$, so after the l'th input: $\frac{j}{j+1} \frac{j+1}{j+2} \cdots \frac{l-1}{l} = \frac{j}{l}$, therefore the total probability of the node $v_j$ to be in the bucket after l inputs is simply $\frac{bucket\_size}{j} \cdot \frac{j}{l} = \frac{bucket\_size}{l}$.
Finally, the probability of choosing any node from the bucket is $\frac{bucket\_size}{l} \cdot \frac{1}{bucket\_size} = \frac{1}{l}$ (we choose uniformly in the bucket), which is exactly the same probability as choosing it from all the $l$ input nodes.
\end{proof}

\subsubsection{Continuous Games: Refreshing the Buffer and Bloom Filter}
The above algorithm works for a single ``round", where we are a node in the network, establishing the known nodes list, and choosing connections for the next round. This is not enough for a continuous game, where nodes in the network come and go frequently, and we need to choose connections repeatedly, using the same buffer. To overcome this difficulty, we can direct the network protocol to propagate live nodes every $\mathcal{T}$ time units, and store two copies of the buffer and filter that reset alternately every $2\mathcal{T}$ time units. This method gives us the ability to remember IPs from a window of $\mathcal{T}$, which is the full available information on the network (as we've assumed honest nodes send their IP address to others at least once every $\mathcal{T}$).

Additionally, in more realistic scenarios in which churn is an issue, and honest nodes may be offline at times, we suggest using larger bucket sizes to preserve a sufficiently large set of alternative connections.

In the full version of the paper, we prove that even if some IPs in the bucket cannot be selected, e.g., if they are stale, selecting uniformly from the remaining nodes in the bucket is equivalent again to a uniform selection (from the whole input nodes).

\subsection{Constant Cost Function}
As a benchmark, and a simple example of the applications of the results we present, we consider the scenario in which there is a constant cost to every node. This is an interesting scenario because in this case the defender doesn't distinguish between any two nodes.

Using the notation of subsection~\ref{subsec::costs}, let $\dif(U) = c \cdot |U|$.
I.e., if the defender will choose the subset of nodes $A$ and the attacker corrupts $B$, then the utility of the game is $c \cdot |B|-\wattack \cdot \delta_{A \subset B}$. 

Due to Lemma~\ref{difficulty_lemma} and the symmetry of the nodes in this example, in all the Nash equilibria in the game, the defender will choose any group of nodes $A \subset V, |A|=H$ with the same probability.
In this case, it is easy to see that there is a single defender's equivalence class. Using Algorithm \ref{remfunc_imp} with $|EG| = 1$ and Theorem~\ref{buffered_game_value}, gives the optimal game value.

For example, let us look at the network with nodes $V = v_1, \cdots, v_l$. For the defender, it is best to uniformly select nodes over all $V$. Using Algorithm \ref{remfunc_imp} with $|EG| = 1, bucket\_size=H$ gives us exactly what we need: $H$ nodes that are uniformly distributed over all the input (same proof as in Lemma \ref{remfunc_implementing_bloom_filter}).

\section{Costs by IP Masks} \label{sec::costs_by_ip_mask}
For the rest of the paper and the restricted-buffer case, we focus on the specific game where the cost function is determined by the IP subnet mask (as described in Subsection~\ref{subsec::costs}). A similar treatment applies to other cost functions.

\begin{lemma} \label{lemma_same_mask_same_prob}
	If $v_1, v_2$ are nodes in the same mask, then $[v_1]_\rightleftharpoons = [v_2]_\rightleftharpoons$. \\
	I.e., for all $A \subset V$ with $|A| = H-1$ and $\forall (\sigma^1, \sigma^2)$ \nash, it holds that $\sigma^1_{A \cup \{v_1\}} = \sigma^1_{A \cup \{v_2\}}$.
\end{lemma}
The proof is easy, and it can be found in the full version of the paper.

The following lemma shows that we can save representatives from each equivalence class. 
\begin{lemma} \label{equivalence_lemma_masks}
	For any defender's strategy	$A = \{v_1, \cdots, v_H\} \in S^1$, it holds that: $[v_1]_\rightleftharpoons \times \cdots \times [v_H]_\rightleftharpoons \subseteq [A]$
	I.e., we can replace any node with a node in the same mask, and still be in the same strategy equivalence class.
\end{lemma}
\begin{proof}
Our cost function does not distinguish between two nodes in the same mask, so we may switch any node with another one in the same mask, and it will cost the same $\Longrightarrow$ Choosing them unequally will ease the attacker's game (according to Lemma \ref{lemma_same_mask_same_prob}) $\Longrightarrow$ We will choose them with the same probability.
\end{proof}

We derive our main practical result as a direct consequence:
\begin{theorem} \label{best_remfunc}
	We can implement the IP mask game on a limited buffer, with the same game value as the limitless-buffer game. 
\end{theorem}
\begin{proof}
	On one hand, we can implement any defender's strategy $[v_1, \cdots, v_H]$ using $[v_1]_\rightleftharpoons, \cdots, [v_H]_\rightleftharpoons$ (Lemma~\ref{equivalence_lemma_masks}). On the other hand, we can implement any $[v]_\rightleftharpoons$ using subnet masks (Lemma~\ref{lemma_same_mask_same_prob}). Finally, by using the masks as buckets, we can implement any strategy in $\hat{S}^1$.
\end{proof}

\subsection{Safety Level} \label{smaller_game_intuition}

Due to many difficulties in completely solving the original game, in this subsection we define the game $\easygameval$ wherein we limit the defender by reducing his available strategies.
Those restrictions give us only a safety level for the original game; however, it makes it easier to compute the strategies.

First, note that the defender should not distinguish between two masks that have the same number of nodes (in any equilibrium, masks with the same number of nodes are selected with the same probability). Therefore, the strategic equivalence classes are defined by mask size.

Denote a few marks: (i) $M_a$ the number of all the nodes in all the masks of size $a$; (ii) $avg_a = \frac{c_{new} + a \cdot c_{node}}{a}$ the average cost of node; (iii) $x_a \in \mathbb{N}$ the number of nodes that the attacker corrupts  in masks of size $a$. Therefore, the attacker invests at least $x_a \cdot avg_a$ money, and the fraction of corrupted nodes is $\frac{x_a}{M_a}$. 

Now, we ease the game by limiting the defender in two ways\footnote{We do so because it is hard to solve the original game.}: He cannot distinguish between two masks with the same size, and he must choose connections independently, each with probability $p_a$.

Denote the random variable $y_a \sim Bin(H, p_a)$ as the defender's choice to choose from mask-size $a$, and  $x'_a$ as $x_a$ if $x_a \ne 0$ or $M_a(1 - \mathbb{E}(y_a))^{\frac{1}{y_a}}$ otherwise. Then 
$\easygameval(\bar{x}, \bar{y}) \ge 
\sum_a x_a \cdot avg_a - \wattack \cdot \prod\limits_a\left( \frac{x'_a}{M_a} \right) ^ {y_a} $, and because $\frac{x'_a}{M_a} \le 1$ we get
$\ge \sum_a x_a \cdot avg_a - \wattack \cdot \prod\limits_a \left( \frac{x'_a}{M_a} \right) $.

I.e., the game value increases linearly with large constants but decreases exponentially with small constants. We will see that these constants make the game unprofitable for any attacker.

%

\subsubsection{Evaluate the value of $\easygameval$} \label{sec::evaluations}
To retrieve an actual value for the calculation above, we examined the behavior of the Bitcoin network. We collected a snapshot of all the nodes that connected to the network (using an honest node in the network and IPs from the site blockchain.info). From this data we collected statistics on the number of masks, and distribution of nodes within each one (see Figures~\ref{buckets_statistics} and \ref{mask_size_to_count}).

\begin{figure}[!ht]
	\centering
	\includegraphics[scale=0.38]{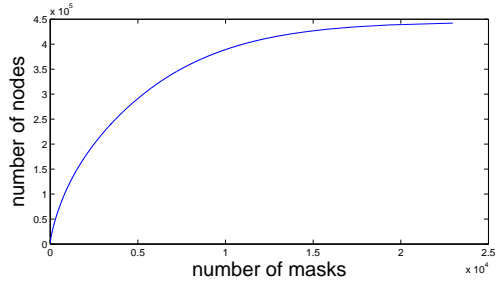}
	\captionsetup{width=.9\linewidth}
	\caption{Number of subnet masks$\backslash16$ in the Bitcoin network (Sep. 2015 data).}
	\label{buckets_statistics}
\end{figure}

\begin{figure}
	\centering
	\includegraphics[scale=0.45]{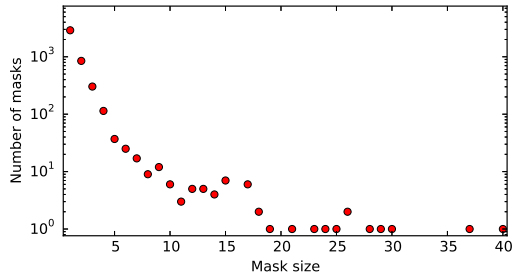}
	\captionsetup{width=.9\linewidth}
	\caption{Number of masks in each mask-size node from a mask.}
	\label{mask_size_to_count}
\end{figure}

We calculated the sizes of each mask size and got:
$ \prod\limits_a \frac{1}{M_a} \approx 8.991 \cdot 10 ^ {-52} $. Therefore the attacker needs to invest huge amounts of resources to gain a negative
 value.

\subsection{Comparing to a Naive Benchmark}
In this section we compare our results to a naive benchmark, where each node chooses its connections uniformly from the buffer, and the buffer is chosen uniformly from the whole network. In the first benchmark (the most naive), we assume that the node may be subject to repeated transmissions of the same IP, which it cannot detect unless that IP is already in the buffer. We assume that the buffer size is the current buffer size of Bitcoin's nodes ($=20480$ unique addr.). For the second benchmark, we assume uniform selection of IPs, but with an accompanying Bloom filter to filter out re-transmissions. We consider a network respectively to the masks in Bitcoin's topology. Considering our strategy, an attacker's response is to create nodes in such a way that $p_a$ is identical for any mask $a$. Then, we calculate the probability of being successfully attacked as a function of the attacker's investment. We choose the value $H=8$ using the number of connections that Bitcoin currently uses by default. The value for $\frac{c_{new}}{c_{node}}$ is deduced using the approximated prices on Amazon EC2. The results are shown in Figure \ref{avg_cost_naive} along with the probability that our own algorithm is successfully attacked. \footnote{The ``spike" in the graph is caused by the number of existing masks (if the attacker corrupts in all masks). In this case, the probability of a successful attack is the same as in the naive approach.}

\begin{figure} 
	\centering
	\includegraphics[width=0.6 \columnwidth]{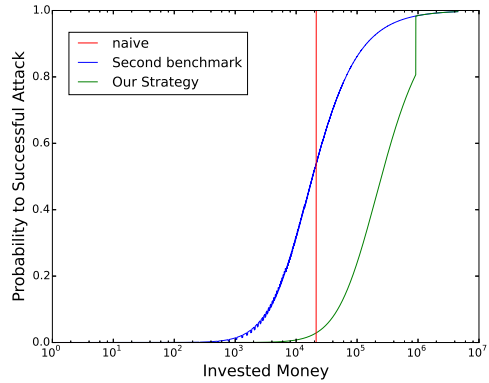}
	\caption{The probability of a successful attack by invested money.}
	\label{avg_cost_naive}
\end{figure}

\section{Future Work \& Conclusions} \label{sec::future_work}
In this work we explored a game theoretic model for P2P network formation. Our results indicate that by using a model for the cost of nodes for an attacker, it is possible to select each peer's connections so as to reduce the likelihood of being isolated by an attacker. We presented an example strategy that can be implemented in any P2P software's code (e.g. Bitcoin, BitTorrent, etc.). We then examined the cost of a successful attack in our model, and compared the results to other strategies.

Future work should extend the model to games that are not zero-sum, and try to better account for attackers attempting to isolate large chunks of the network (simulations we conducted show that the strategies that protect single nodes are also good for the network as a whole). As a motivated example, splitting the network into chunks in Bitcoin may disrupt the proportions of the computational power of the attacker and ease double-spending attacks.

This paper uses a single, existing network utility for the cost function (the IP mask) to ease the implementation on current applications. This is not mandatory, and there is more information about peers that we can utilize such as round-trip time, traceroute, and other network features. Another possibility is to augment nodes with additional information that will be hard for an attacker to fake. For example, it may be interesting to examine the proof-of-work concept as something that increases the reliability of a peer. 

Finally, we used a specific strategy and compared it to other benchmarks. Finding a way to solve the game mathematically (without assumptions or heuristic strategies), or finding a defender's strategy that is in some \nash, will create the optimal algorithm for any mask-based P2P system. This creates an optimal solution that we can all hope for within this model.

\bibliographystyle{LNCS/splncs03}
\bibliography{pyf}

\end{document}